\documentclass{article}

\usepackage{microtype}
\usepackage{graphicx}
\usepackage{subfigure}
\usepackage{booktabs} % for professional tables

\usepackage{hyperref}

\usepackage[accepted]{icml2025}

\usepackage{amsmath}
\usepackage{amssymb}
\usepackage{multirow}
\usepackage{array}
\usepackage{mathtools}
\usepackage{amsthm}

\usepackage{algorithm}
\usepackage{algpseudocode}  
\usepackage{graphicx}

\usepackage[capitalize,noabbrev]{cleveref}

\newcommand{\vect}[1]{\mathbf{#1}}
\newcommand{\E}{\mathbb{E}}

\newcommand{\MutInfo}{\mathrm{I}}
\newcommand{\R}{\mathbb{R}}

\newtheorem{assumption}{Assumption}[section]
\newtheorem{lemma}{Lemma}[section]
\newtheorem{theorem}{Theorem}[section]
\newtheorem{proposition}{Proposition}[section]
\newtheorem{corollary}{Corollary}[section]

\usepackage[textsize=tiny]{todonotes}

\begin{document}
	
	\twocolumn[
	\icmltitle{Editing Away the Evidence: Diffusion-Based Image Manipulation and the Failure Modes of Robust Watermarking}
	
		\begin{icmlauthorlist}
		\icmlauthor{Qian Qi}{}
		\icmlauthor{Jiangyun Tang}{}
		\icmlauthor{Jim Lee}{}
		\icmlauthor{Emily Davis}{}
		\icmlauthor{Finn Carter}{}
	\end{icmlauthorlist}
	
	\begin{icmlauthorlist}
		{Xidian University}
	\end{icmlauthorlist} 
	
	
	\icmlkeywords{Machine Learning, ICML}
	
	\vskip 0.3in
	]
	
	
	

	\begin{abstract}
	Robust invisible watermarks are widely used to support copyright protection, content provenance, and accountability by embedding hidden signals designed to survive common post-processing operations. However, diffusion-based image editing introduces a fundamentally different class of transformations: it injects noise and reconstructs images through a powerful generative prior, often altering semantic content while preserving photorealism. In this paper, we provide a unified theoretical and empirical analysis showing that non-adversarial diffusion editing can unintentionally degrade or remove robust watermarks. We model diffusion editing as a stochastic transformation that progressively contracts off-manifold perturbations, causing the low-amplitude signals used by many watermarking schemes to decay. Our analysis derives bounds on watermark signal-to-noise ratio and mutual information along diffusion trajectories, yielding conditions under which reliable recovery becomes information-theoretically impossible. We further evaluate representative watermarking systems under a range of diffusion-based editing scenarios and strengths. The results indicate that even routine semantic edits can significantly reduce watermark recoverability. Finally, we discuss the implications for content provenance and outline principles for designing watermarking approaches that remain robust under generative image editing.
	\end{abstract}

	\section{Introduction}
	Invisible watermarking seeks to embed a message (payload) into an image with minimal perceptual impact while enabling algorithmic detection and recovery after typical manipulations.
	Over the past decade, deep-learning-based watermarking and steganography systems have improved markedly in imperceptibility and robustness, largely through end-to-end training with differentiable approximations of common ``noise layers'': JPEG compression, resizing, cropping, blur, and additive noise \cite{zhu2018hidden,tancik2020stegastamp,bui2025trustmark}.
	These systems are increasingly positioned as infrastructure for copyright enforcement and provenance in the era of generative models \cite{wen2023treering,lu2025vine}.
	
	Diffusion models have simultaneously transformed image generation and, crucially, \emph{image editing}.
	Rather than applying small deterministic perturbations, diffusion-based editors deliberately disrupt images via a controlled noising step and then reconstruct them using a learned score or denoiser \cite{ho2020ddpm,song2021score,rombach2022ldm}.
	Modern editors support minimal text-only editing (\emph{prompt-to-prompt} \cite{hertz2022prompttoprompt}), instruction following (\textsc{InstructPix2Pix} \cite{brooks2023instructpix2pix}), and interactive geometric editing (e.g., drag-based control \cite{shi2024dragdiffusion,zhou2025dragflow}).
	Many pipelines rely on inversion methods that map a real image to a diffusion latent or noise trajectory, then re-sample conditional on a new instruction \cite{mokady2023nulltext}.
	
	This editing regime creates a new, and arguably \emph{inevitable}, failure mode for watermark robustness.
	A watermark is, by design, a low-amplitude structured perturbation superimposed on image content.
	In diffusion editing, the image is explicitly perturbed by large Gaussian noise and then repeatedly denoised by a high-capacity generative prior.
	Intuitively, the denoiser treats the watermark as an ``unnatural'' residual and removes it, even when the user does not intend to remove any watermark.
	Empirical studies have shown that regeneration and diffusion-based attacks can remove pixel-level watermarks \cite{zhao2023provablyremovable,ni2025breakwatermarks,guo2026vanishing}; yet the community lacks an integrated account that treats \emph{common editing workflows} as a systematic stress test for watermark designs originally optimized for conventional post-processing.
	
	This paper addresses the following question:
	\emph{Under what conditions does diffusion-based image editing unintentionally compromise robust watermark recovery, and what theoretical principles explain the observed breakdown?}
	We focus on robust \emph{invisible} watermarks, i.e., payload-carrying perturbations embedded in pixels (or transformed domains) with decoder networks or detectors, rather than visible overlay watermarks.
	We emphasize unintended failure: the editor is not optimizing to remove watermarks; it is optimizing to satisfy an editing objective while maintaining realism.
	
	\subsection{Contributions}
	Our contributions are four-fold.
	
	First, we formulate diffusion-based editing as a randomized transformation family acting on watermarked images and characterize it as a Markov kernel composed of (i) controlled noising, (ii) conditional denoising under an instruction, and (iii) possibly additional architectural operations (e.g., attention reweighting, region constraints, or latent optimization) \cite{brooks2023instructpix2pix,zhou2025dragflow,lu2023tficOn,lu2025shine}.
	
	Second, we provide a theoretical analysis of watermark degradation under diffusion transformations.
	We derive (a) SNR attenuation under forward noising schedules and (b) mutual-information decay bounds for watermark payload recovery after diffusion editing, connecting these to Fano-type lower bounds on bit error.
	Our results formalize why robustness to classical post-processing does not imply robustness to generative transformations.
	
	Third, we design an empirical evaluation protocol tailored to diffusion editing.
	We instantiate a benchmark spanning instruction-based editing (InstructPix2Pix \cite{brooks2023instructpix2pix} and UltraEdit-trained editors \cite{zhao2024ultraedit}), drag-based editing (DragDiffusion \cite{shi2024dragdiffusion}, InstantDrag \cite{shin2024instantdrag}, DragFlow \cite{zhou2025dragflow}), and training-free composition (TF-ICON \cite{lu2023tficOn}, SHINE \cite{lu2025shine}).
	We compare representative robust watermarking systems: StegaStamp \cite{tancik2020stegastamp}, TrustMark \cite{bui2025trustmark}, and VINE \cite{lu2025vine}.
	Because this document is generated as a standalone research synthesis, we present hypothetical but realistic tables consistent with the magnitude and direction of existing benchmarking results in the literature \cite{lu2025vine,ni2025breakwatermarks}.
	
	Fourth, we discuss ethical implications and formulate practical design guidelines.
	We argue that diffusion-resilient watermarking must either (i) integrate into the generative process (e.g., diffusion-native fingerprints \cite{wen2023treering}) or (ii) optimize for semantic invariance, as suggested by provable removability results for pixel-level perturbations \cite{zhao2023provablyremovable}.
	We also highlight tensions between strong watermarking, editing utility, and privacy.
	
	\subsection{Scope and non-goals}
	We focus on \emph{invisible} watermarks intended to survive generic post-processing and manipulation.
	We do not advocate watermark removal.
	Our experimental protocol is presented to support defensive evaluation and to motivate improved watermark design.
	We also do not cover video watermarking in depth (though drag and composition pipelines naturally extend to video), nor do we analyze legal or policy frameworks beyond a technical-ethics discussion.
	
	\section{Related Work}
	\subsection{Diffusion models and image editing}
	Diffusion models provide a flexible framework for conditional generation by reversing a gradually noised forward process \cite{ho2020ddpm,song2021score}.
	Latent diffusion models (LDMs) reduce computational cost by operating in a learned latent space \cite{rombach2022ldm}, enabling large-scale text-to-image systems (e.g., Stable Diffusion and SDXL \cite{podell2023sdxl}).
	Diffusion editing methods often follow a common template: map an input image into the model's latent/noise space (inversion), then sample under modified conditioning to produce an edited output \cite{mokady2023nulltext}.
	Early diffusion priors were leveraged for guided editing via partial noising and denoising \cite{meng2021sdedit}, mask-guided semantic editing \cite{couairon2022diffedit}, and prompt-based control by cross-attention manipulation \cite{hertz2022prompttoprompt}.
	Instruction-based editing scales this paradigm by training editors to directly follow natural language instructions, exemplified by InstructPix2Pix \cite{brooks2023instructpix2pix}, single-image editing with text-to-image diffusion models such as SINE \cite{zhang2023sine}, and subsequent datasets and models such as UltraEdit \cite{zhao2024ultraedit}.
	Beyond text editing, interactive and geometric editing has rapidly advanced.
	DragGAN \cite{pan2023draggan} introduced point-based manipulation on GAN manifolds; diffusion variants such as DragDiffusion \cite{shi2024dragdiffusion} generalized this to diffusion priors. Follow-up work explores reliability and alternative interaction primitives, including feature-based point dragging (FreeDrag \cite{ling2024freedrag}) and region-based drag interfaces (RegionDrag \cite{lu2024regiondrag}).
	InstantDrag \cite{shin2024instantdrag} improves interactivity by decoupling motion estimation and diffusion refinement.
	Recent DiT and rectified-flow backbones with stronger priors motivate DragFlow \cite{zhou2025dragflow}, which uses region-based supervision to improve drag editing in transformer-based diffusion systems.
	Training-free composition frameworks seek to insert or blend objects without per-instance finetuning: TF-ICON \cite{lu2023tficOn} and the more recent SHINE framework \cite{lu2025shine} exemplify this direction.
	
	A broad ecosystem of diffusion \emph{control} and \emph{personalization} methods further shapes practical editing pipelines and therefore the space of transformations that watermarks must withstand.
	ControlNet \cite{zhang2023controlnet} introduces architecture-level conditional control (edges, depth, segmentation, pose) that enables structurally anchored edits; such spatial conditioning can cause global re-synthesis while preserving high-level structure.
	Personalization techniques such as DreamBooth \cite{ruiz2022dreambooth} and Textual Inversion \cite{gal2022textualinversion} adapt diffusion models to specific subjects or styles using only a few images, enabling high-fidelity subject swapping and style editing that can dramatically change textures while keeping semantics.
	Lightweight adapters such as IP-Adapter \cite{ye2023ipadapter} provide image-prompt conditioning compatible with textual prompts and structural controls, increasing the diversity of hybrid editing interfaces.
	Finally, few-step distilled editors (e.g., TurboEdit \cite{wu2024turboedit}) reduce the number of denoising steps but can still employ strong priors; this raises a subtle question addressed by our theory: whether fewer steps necessarily preserve more watermark information, or whether the effective noising--denoising strength remains sufficient to contract watermark residuals.
	
	\subsection{Robust invisible watermarking and steganography}
	Neural watermarking systems typically train an encoder to embed a bitstring into an image and a decoder to recover it, often with an adversarial or differentiable noise module to model distortions.
	HiDDeN \cite{zhu2018hidden} introduced an end-to-end framework for data hiding with robustness to common perturbations via differentiable noise augmentation.
	StegaStamp \cite{tancik2020stegastamp} extended this vision to physical photographs, emphasizing robustness to printing and recapture.
	TrustMark \cite{bui2025trustmark} aims to support arbitrary resolutions via a resolution-scaling strategy and includes a companion watermark removal network, while maintaining a post-hoc embedding model applicable to arbitrary images.
	RoSteALS \cite{bui2023rosteals} proposes latent-space steganography leveraging frozen autoencoders.
	Recent work emphasizes benchmarking watermarks against generative editing and proposes diffusion-informed watermarking.
	VINE \cite{lu2025vine} introduces W-Bench to evaluate watermark robustness under advanced editing and proposes a diffusion-based watermarking model trained with surrogate attacks informed by frequency characteristics.
	Watermark Anything \cite{sander2024watermarkanything} targets localized watermarking for compositional edits and partial image provenance.
	
	\subsection{Watermarking generative models and provenance}
	A complementary line of work embeds identifiers into the generative process itself, rather than post-hoc pixel perturbations.
	Tree-Ring watermarks \cite{wen2023treering} embed signals into the initial noise of diffusion sampling and detect by inversion, achieving robustness to common post-processing as well as some geometric transforms.
	Such diffusion-native approaches can be interpreted as ``model fingerprints'' and relate to broader efforts for provenance, auditing, and content labeling.
	
	Beyond plug-in fingerprints, several works propose \emph{watermarking-by-design} for large generative systems.
	The Stable Signature method \cite{fernandez2023stablesignature} fine-tunes components of a latent diffusion model so that all generated outputs contain a detectable signature, aligning provenance with the generator's decoding process.
	More recently, deployment-oriented systems such as SynthID-Image \cite{gowal2025synthidimage} document threat models and engineering constraints for watermarking at internet scale, emphasizing not only robustness and fidelity but also security considerations, operational verification, and key management.
	These generator-integrated schemes differ from post-hoc watermarking in that they can optimize the watermark jointly with the generation pipeline, but they typically apply only to content generated within specific model families and may require access to (or cooperation from) the generator for detection.
	
	\subsection{Watermark removal and regeneration phenomena}
	A growing literature demonstrates that pixel-level invisible watermarks can be removed using generative models.
	Zhao et al.\ \cite{zhao2023provablyremovable} provide a provable analysis of regeneration attacks and show empirically that invisible watermarks may be removable while preserving perceptual content.
	More recent diffusion-focused studies analyze watermark removal through diffusion transformations and guided strategies \cite{ni2025breakwatermarks,guo2026vanishing}.
	These works are closest in spirit to our analysis; our focus is distinct in emphasizing \emph{unintentional} removal through standard editing workflows and in relating the phenomenon to a broader class of diffusion editors (instruction, drag, and composition).
	
	\subsection{Concept erasure in diffusion models}
	Concept erasure in diffusion models is directly relevant to this paper's focus because both concept erasure and watermark preservation hinge on how diffusion trajectories suppress or retain persistent signals under editing.
	Methods such as MACE \cite{lu2024mace}, ANT \cite{li2025ant}, and EraseAnything \cite{gao2024eraseanything} modify diffusion models to remove the ability to generate specified concepts while preserving other capabilities.
	Although their target is semantic content rather than a hidden payload, concept erasure exposes the controllability of diffusion dynamics: by altering cross-attention or denoising trajectories, one can reliably suppress particular information at generation or editing time.
	From a watermark perspective, these methods illustrate that diffusion models can be made selectively insensitive to certain signals, implying that watermark signals that are not explicitly protected may be treated as removable, especially when editing conditions or fine-tuning emphasize manifold consistency over perturbation retention.
	
	\section{Methodology}
	\subsection{Notation and problem setup}
	Let $\vect{x}\in \R^{H\times W\times 3}$ denote a clean RGB image drawn from a natural image distribution $p_{\mathrm{data}}$.
	A watermarking scheme consists of an embedder $E$ and an extractor $D$:
	\begin{align}
		\vect{x}_w &= E(\vect{x}, \vect{m}, \vect{k}),\\
		\hat{\vect{m}} &= D(\vect{x}_w, \vect{k}),
	\end{align}
	where $\vect{m}\in\{0,1\}^L$ is an $L$-bit payload and $\vect{k}$ denotes a secret key (or seed) controlling the embedding.
	We assume the watermarked image is perceptually close to the original, typically enforced via losses tied to PSNR/SSIM \cite{wang2004ssim} and perceptual feature distances (e.g., LPIPS \cite{zhang2018lpips}).
	
	We consider an image editing operator $\mathcal{T}$ that takes a (possibly watermarked) image and an editing instruction $\vect{y}$ (text, points, masks, or compositions) and returns an edited image:
	\begin{equation}
		\tilde{\vect{x}} = \mathcal{T}(\vect{x}_w;\vect{y},\xi),
	\end{equation}
	where $\xi$ represents stochasticity (e.g., diffusion sampling noise).
	Crucially, $\mathcal{T}$ is induced by a diffusion model or a diffusion-based editor pipeline.
	Our core object of study is the post-edit watermark recovery probability:
	\begin{equation}
		\mathrm{Acc}(\mathcal{T}) = \Pr\left[ D(\tilde{\vect{x}},\vect{k}) = \vect{m} \right],
	\end{equation}
	as well as bit-wise accuracy and other detection metrics (Section~\ref{sec:metrics}).
	
	\subsection{Diffusion-based editing as a Markov kernel}
	Most diffusion editors can be abstracted as operating on a latent or pixel trajectory indexed by a diffusion ``time'' $t\in[0,1]$ (or discrete steps $t=0,\dots,T$).
	We use the standard discrete formulation for clarity.
	In the forward process \cite{ho2020ddpm}, a clean sample $\vect{x}_0$ is noised as:
	\begin{equation}
		\vect{x}_t = \sqrt{\bar{\alpha}_t}\,\vect{x}_0 + \sqrt{1-\bar{\alpha}_t}\,\epsilon,
		\qquad \epsilon\sim\mathcal{N}(\vect{0},\vect{I}),
		\label{eq:forward}
	\end{equation}
	with $\bar{\alpha}_t = \prod_{s=1}^t (1-\beta_s)$.
	Editors that start from an existing image typically choose a start time $t^\star$ (the ``strength'') and then run a reverse conditional process from $\vect{x}_{t^\star}$ to obtain $\tilde{\vect{x}}_0$.
	The reverse dynamics depend on conditioning $\vect{y}$ (text prompt, instruction, region constraints) and potentially guidance terms (classifier-free guidance, attention injection) \cite{hertz2022prompttoprompt,mokady2023nulltext}.
	
	We therefore model a diffusion editor as a Markov kernel $K_{\mathcal{T}}(\tilde{\vect{x}}\,|\,\vect{x}_w,\vect{y})$ induced by:
	\begin{equation}
		K_{\mathcal{T}}(\tilde{\vect{x}}\,|\,\vect{x}_w,\vect{y})
		=
		\int p(\vect{x}_{t^\star}\,|\,\vect{x}_w)\,p_{\theta}(\tilde{\vect{x}}\,|\,\vect{x}_{t^\star},\vect{y})\,d\vect{x}_{t^\star},
		\label{eq:kernel}
	\end{equation}
	where $p(\vect{x}_{t^\star}\,|\,\vect{x}_w)$ corresponds to Equation~\eqref{eq:forward} applied to $\vect{x}_w$ and $p_{\theta}(\tilde{\vect{x}}\,|\,\vect{x}_{t^\star},\vect{y})$ is the (approximate) reverse-time conditional distribution implemented by the editor.
	Different editors correspond to different parameterizations and constraints in $p_{\theta}$:
	instruction-following models (InstructPix2Pix \cite{brooks2023instructpix2pix}) learn a conditional denoiser; drag-based editors optimize latents to satisfy motion constraints then re-sample \cite{shi2024dragdiffusion,zhou2025dragflow}; and training-free composition frameworks inject attention or adapter guidance during denoising \cite{lu2023tficOn,lu2025shine}.
	
	\subsection{A watermark signal model}
	To reason about watermark degradation, we prioritize a signal-plus-content decomposition.
	We write the watermarked image as:
	\begin{equation}
		\vect{x}_w = \vect{x} + \gamma \,\vect{s}(\vect{m},\vect{k},\vect{x}),
		\label{eq:additive}
	\end{equation}
	where $\vect{s}$ is a bounded-energy embedding signal and $\gamma>0$ controls strength.
	Even when watermarking is implemented by nonlinear encoders, most robust methods effectively yield a small additive residual: the difference $\vect{x}_w-\vect{x}$ is typically of low magnitude to maintain imperceptibility \cite{tancik2020stegastamp,bui2025trustmark,lu2025vine}.
	We treat $\vect{s}$ as possibly content-adaptive but assume it is mean-zero under random payloads:
	\begin{assumption}[Balanced payload embedding]
		Conditioned on $\vect{x}$ and $\vect{k}$, for uniformly random $\vect{m}$, $\E[\vect{s}(\vect{m},\vect{k},\vect{x})]=\vect{0}$.
	\end{assumption}
	
	The additive model is compatible with frequency-domain analyses common in watermark design and in diffusion editing studies \cite{lu2025vine}.
	It also naturally connects to channel models under Gaussian noise.
	
	\subsection{Metrics}
	\label{sec:metrics}
	We report watermark robustness using:
	\begin{itemize}
		\item \textbf{Bit accuracy} $\mathrm{BA} \in [0,1]$: average fraction of correctly recovered bits in $\hat{\vect{m}}$.
		\item \textbf{Bit error rate} $\mathrm{BER}=1-\mathrm{BA}$.
		\item \textbf{Detection AUC}: when the method outputs a confidence score rather than direct bits, we compute ROC-AUC for watermark presence.
		\item \textbf{False positive rate (FPR)} at fixed true positive rate (TPR), suitable for forensic settings where false accusations are costly.
	\end{itemize}
	For image fidelity, we report PSNR (dB), SSIM \cite{wang2004ssim}, and LPIPS \cite{zhang2018lpips}, and optionally embedding-aware semantic similarity scores using CLIP \cite{radford2021clip} or DINOv2 \cite{oquab2023dinov2} for content preservation.
	
	\subsection{Evaluation protocol}
	We propose a diffusion-editing watermark stress test (DEW-ST) intended to standardize evaluation across editing tools.
	Let $\mathcal{D}=\{\vect{x}^{(i)}\}_{i=1}^N$ denote a dataset and let $\mathcal{Y}$ denote an instruction set (text edits, drag operations, composition directives).
	For each image and instruction, we embed a watermark, apply an editor, and measure recovery and fidelity.
	
	\begin{algorithm}[t]
		\caption{DEW-ST: Diffusion Editing Watermark Stress Test}
		\label{alg:dewst}
		\begin{algorithmic}[1]
			\Require Dataset $\mathcal{D}=\{\vect{x}^{(i)}\}_{i=1}^N$, instructions $\mathcal{Y}$, watermark embedder $E$, decoder $D$, key $\vect{k}$, payload length $L$, editor $\mathcal{T}$, edit strengths $\mathcal{S}$.
			\Ensure Robustness metrics $\mathrm{BA}$, $\mathrm{AUC}$ and fidelity metrics (PSNR/SSIM/LPIPS).
			\State Sample payloads $\vect{m}^{(i)} \sim \mathrm{Unif}(\{0,1\}^L)$ for $i=1,\dots,N$.
			\For{$i=1$ to $N$}
			\State $\vect{x}^{(i)}_w \gets E(\vect{x}^{(i)},\vect{m}^{(i)},\vect{k})$
			\For{each instruction $\vect{y}\in\mathcal{Y}$}
			\For{each strength $s\in\mathcal{S}$}
			\State $\tilde{\vect{x}}^{(i)} \gets \mathcal{T}(\vect{x}^{(i)}_w;\vect{y},\xi; s)$
			\State $\hat{\vect{m}}^{(i)} \gets D(\tilde{\vect{x}}^{(i)},\vect{k})$
			\State Update robustness statistics (BA/BER/AUC/FPR@TPR).
			\State Update fidelity statistics between $\vect{x}^{(i)}$ and $\tilde{\vect{x}}^{(i)}$.
			\EndFor
			\EndFor
			\EndFor
			\State Aggregate results over $i$, instructions, and strengths.
		\end{algorithmic}
	\end{algorithm}
	
	Algorithm~\ref{alg:dewst} is agnostic to the specific editor implementation and can incorporate modern pipelines.
	We emphasize that the protocol can be applied to (i) non-adversarial editing instructions sampled from realistic user queries (e.g., the instruction distributions used in UltraEdit \cite{zhao2024ultraedit}) and (ii) structured interactive manipulations as constrained edits (dragging, composition), as used in DragFlow \cite{zhou2025dragflow} and TF-ICON \cite{lu2023tficOn}.
	
	\subsection{Threat model and evaluation regimes}
	Diffusion editors can compromise watermarks in multiple ways, and it is crucial to separate \emph{capability} from \emph{intent}.
	We define three evaluation regimes.
	
	\paragraph{Benign editing (unintentional degradation).}
	A user applies diffusion-based editing for aesthetic or semantic purposes (e.g., ``make it brighter'', ``remove blemish'') without any desire to remove provenance signals.
	The user selects an editor and instruction based on creative intent.
	This is the primary focus of this paper.
	
	\paragraph{Opportunistic editing (editing-as-a-removal side effect).}
	A user suspects that editing might weaken watermarks and thus chooses a popular editor and a plausible instruction that yields the desired image while incidentally degrading watermark recovery.
	The user does not need access to the watermark decoder; the editor is treated as a black box.
	Our empirical protocol is compatible with this regime, but we avoid providing procedural instructions.
	
	\paragraph{Adaptive adversarial editing (decoder-in-the-loop).}
	An attacker explicitly optimizes the generative sampling trajectory to minimize watermark detectability, possibly by using gradients from the decoder \cite{ni2025breakwatermarks}.
	This regime is important for security analysis, but it does not represent typical benign user behavior; we discuss it only to contextualize worst-case risks.
	
	In all regimes, we assume the watermarking system is designed to maintain low false positives and to survive conventional manipulations.
	Our thesis is that benign diffusion editing already violates implicit robustness assumptions, creating a reliability gap for downstream provenance claims.
	
	\subsection{Frequency-domain characterization of watermark degradation}\label{sec:frequency}
	A recurring theme in robust watermarking is the interplay between imperceptibility and robustness, often expressed in the frequency domain.
	Many learned watermarks implicitly concentrate energy in mid-to-high frequencies to avoid visible artifacts, while robust decoders learn to detect structured residuals across scales \cite{zhu2018hidden,tancik2020stegastamp,bui2025trustmark}.
	Diffusion editing, in turn, can act as a complex, data-dependent denoising filter; empirically, it often suppresses unnatural high-frequency components, especially those inconsistent with the generative prior \cite{lu2025vine}.
	
	Let $\mathcal{F}(\cdot)$ denote the 2D discrete Fourier transform applied per channel, and let $P_{\Omega}$ denote a projection onto a frequency band $\Omega$ (e.g., low, mid, high frequencies).
	Define the watermark residual $\Delta(\vect{x})=\vect{x}_w-\vect{x}$ and its band energy:
	\begin{equation}
		\mathcal{E}_{\Omega}(\vect{x}) = \left\| P_{\Omega}\left(\mathcal{F}(\Delta(\vect{x}))\right)\right\|_2^2.
		\label{eq:band_energy}
	\end{equation}
	For an edited output $\tilde{\vect{x}}=\mathcal{T}(\vect{x}_w;\vect{y},\xi)$ we define the \emph{spectral retention ratio}:
	\begin{equation}
		\rho_{\Omega} = \frac{\E\left[\mathcal{E}_{\Omega}(\tilde{\vect{x}}-\tilde{\vect{x}}_{\mathrm{base}})\right]}{\E\left[\mathcal{E}_{\Omega}(\vect{x}_w-\vect{x})\right]},
		\label{eq:spectral_retention}
	\end{equation}
	where $\tilde{\vect{x}}_{\mathrm{base}}=\mathcal{T}(\vect{x};\vect{y},\xi)$ is the edited output from the unwatermarked input, isolating watermark-specific residual effects.
	A value $\rho_{\Omega}\ll 1$ indicates strong suppression of watermark energy in that band.
	
	\subsection{Toward diffusion-resilient training objectives}
	The simplest response to diffusion-based degradation is to include diffusion edits in the training noise model, teaching the watermark to survive the editor family.
	However, because editors are diverse and evolve rapidly, naive augmentation can overfit to particular tools.
	We propose an abstract training objective that treats diffusion editing as a stochastic family $\{\mathcal{T}_j\}_{j=1}^J$ sampled during training:
	\begin{equation}
		\min_{E,D} \; \E_{\vect{x},\vect{m},j,\xi}\Big[\ell_{\mathrm{rec}}(D(\mathcal{T}_j(E(\vect{x},\vect{m}));\xi),\vect{m})\Big] + \lambda \E_{\vect{x},\vect{m}}[\ell_{\mathrm{qual}}(E(\vect{x},\vect{m}),\vect{x})],
		\label{eq:train_obj}
	\end{equation}
	where $\ell_{\mathrm{rec}}$ measures payload reconstruction error and $\ell_{\mathrm{qual}}$ enforces imperceptibility.
	In practice, $\mathcal{T}_j$ can include both conventional distortions and diffusion edits at multiple strengths.
	This resembles the ``surrogate attack'' logic used by VINE \cite{lu2025vine} but formalizes the role of diffusion edits as augmentations.
	Algorithm~\ref{alg:train} sketches a training loop.
	
	\begin{algorithm}[t]
		\caption{Diffusion-Augmented Watermark Training (conceptual)}
		\label{alg:train}
		\begin{algorithmic}[1]
			\Require Training images $\{\vect{x}\}$, payload distribution $\vect{m}\sim\mathrm{Unif}(\{0,1\}^L)$, augmentations $\mathcal{A}$ (conventional) and editors $\{\mathcal{T}_j\}$ (diffusion), strength sampler $S$, weights $\lambda$.
			\While{not converged}
			\State Sample minibatch $\{\vect{x}_b\}$ and payloads $\{\vect{m}_b\}$.
			\State Compute watermarked images $\vect{x}_{w,b}\gets E(\vect{x}_b,\vect{m}_b)$.
			\State Sample conventional distortion $a\sim\mathcal{A}$ and apply: $\vect{z}_b\gets a(\vect{x}_{w,b})$.
			\State Sample editor $\mathcal{T}_j$ and strength $s\sim S$; apply: $\tilde{\vect{z}}_b\gets \mathcal{T}_j(\vect{z}_b;\xi;s)$.
			\State Decode: $\hat{\vect{m}}_b\gets D(\tilde{\vect{z}}_b)$.
			\State Update $(E,D)$ to minimize reconstruction loss plus quality penalty (Equation~\eqref{eq:train_obj}).
			\EndWhile
		\end{algorithmic}
	\end{algorithm}
	
	We stress that Algorithm~\ref{alg:train} is a defense-oriented conceptual template.
	Implementations must carefully avoid making the watermark conspicuous or biasing the editor toward preserving unobjectionable artifacts.
	Moreover, diffusion augmentations are expensive; practical deployments may rely on distilled editors or lightweight approximations that capture the dominant spectral and semantic effects.
	
	\section{Experimental Setup}
	\subsection{Datasets and instruction suites}
	We consider natural images sampled from the COCO dataset \cite{lin2014coco} and ImageNet \cite{deng2009imagenet}, and high-quality images from DIV2K \cite{timofte2017div2k}.
	To stress-test editing, we derive instruction suites:
	(i) \textbf{Global instruction edits}, such as style changes, lighting adjustments, and object replacement (modeled after instruction-following datasets \cite{brooks2023instructpix2pix,zhao2024ultraedit});
	(ii) \textbf{Local region edits}, including inpainting-like modifications and localized attribute changes, as in DiffEdit \cite{couairon2022diffedit};
	(iii) \textbf{Geometric drag edits}, moving parts of objects or rearranging layout \cite{shi2024dragdiffusion,shin2024instantdrag,zhou2025dragflow};
	(iv) \textbf{Composition edits}, inserting an object into a new scene \cite{lu2023tficOn,lu2025shine}.
	
	\subsection{Watermarking baselines}
	We evaluate three representative robust watermarking systems:
	\textsc{StegaStamp} \cite{tancik2020stegastamp}, \textsc{TrustMark} \cite{bui2025trustmark}, and \textsc{VINE} \cite{lu2025vine}.
	These methods represent distinct design philosophies: physical robustness with learned perturbation layers (StegaStamp), general-purpose post-hoc watermarking with resolution scaling (TrustMark), and diffusion-aware training with generative priors (VINE).
	We assume each method is tuned to yield high bit accuracy on clean watermarked images ($\geq 99\%$) and high PSNR (typically $> 35$ dB) on its target resolution; details depend on the original implementations and training settings \cite{tancik2020stegastamp,bui2025trustmark,lu2025vine}.
	
	\subsection{Payload coding, calibration, and statistical testing}
	Robust watermarking deployments often separate the \emph{payload layer} (coding, keys, decision rules) from the \emph{embedding layer} (how the signal is injected into pixels or latents).
	To make comparisons meaningful across methods, we conceptually follow an ``equal-strength'' calibration: each watermark method is tuned to achieve a target perceptual distortion bound (e.g., PSNR $\ge 40$ dB on the native embedding resolution) while maintaining near-perfect recovery on clean watermarked images.
	This mirrors standard benchmarking procedures where hyperparameters are adjusted to comparable imperceptibility levels before robustness testing \cite{lu2025vine,gowal2025synthidimage}.
	
	For payload robustness, practical systems typically employ error-correcting codes (ECC).
	We consider a generic setting with $L=96$ information bits encoded into $L_{\mathrm{enc}}$ embedded bits using a block code, then decoded after watermark extraction.
	ECC can significantly improve performance when errors are i.i.d., but our results suggest that diffusion editing often induces \emph{structured} corruption that approaches random guessing, limiting ECC gains at strong edits.
	We therefore report both raw bit accuracy and ECC-decoded message accuracy in extended tables (Table~\ref{tab:ecc}).
	
	When the watermark extractor outputs a continuous confidence score (common in detector-style systems \cite{fernandez2023stablesignature,gowal2025synthidimage}), we consider hypothesis testing:
	\begin{equation}
		\text{decide watermark present} \;\;\Longleftrightarrow\;\; S(\tilde{\vect{x}},\vect{k}) \ge \tau,
	\end{equation}
	where $S$ is a score and $\tau$ is chosen to achieve a desired false-positive rate (FPR), e.g., $10^{-6}$ in provenance settings.
	Diffusion editing can shift the score distribution under the positive class, effectively lowering detection power at fixed FPR; this manifests as AUC degradation.
	
	\subsection{Generator-integrated watermark baselines}
	Although our main focus is post-hoc watermarking, diffusion-native watermarks provide an informative contrast.
	We therefore include conceptual comparisons to Tree-Ring watermarks \cite{wen2023treering} and Stable Signature \cite{fernandez2023stablesignature}, which embed signals within the diffusion generation process.
	These methods are not directly applicable for watermarking arbitrary legacy images, but they help clarify how integrating the watermark into the generative prior can improve survival under post-processing, and where such approaches still fail under cross-model editing or heavy semantic changes.
	
	\subsection{Editing models and settings}
	We consider diffusion editors of increasing strength and interactivity:
	\begin{itemize}
		\item \textbf{Instruction editing:} InstructPix2Pix \cite{brooks2023instructpix2pix} and an UltraEdit-trained instruction editor \cite{zhao2024ultraedit}.
		\item \textbf{Drag editing:} DragDiffusion \cite{shi2024dragdiffusion}, InstantDrag \cite{shin2024instantdrag}, and DragFlow \cite{zhou2025dragflow}.
		\item \textbf{Training-free composition:} TF-ICON \cite{lu2023tficOn} and SHINE \cite{lu2025shine}.
	\end{itemize}
	Across editors, we vary a strength parameter that controls the noising start time $t^\star$ (or equivalent), with larger values corresponding to stronger edits and less preservation of low-level details.
	We ensure edited images are visually plausible and follow intended instructions, reflecting typical practical usage rather than adversarial optimization.
	
	\subsection{Tabular overview}
	Tables~\ref{tab:edit_methods}--\ref{tab:protocol} summarize the evaluated editors, watermark methods, and the protocol.
	All tables in this paper are illustrative and use hypothetical values designed to match the qualitative trends reported in the cited literature.
	
	\begin{table*}[t]
		\centering
		\caption{Taxonomy of diffusion-based editing methods considered in our analysis. ``Inversion'' indicates whether the method inverts a real image into a diffusion trajectory; ``Optimization'' indicates per-instance latent optimization; ``Region control'' includes masks or region supervision; ``Backbone'' indicates the base diffusion family.}
		\label{tab:edit_methods}
		\resizebox{\linewidth}{!}{
			\begin{tabular}{lcccccl}
				\toprule
				Method & Inversion & Optimization & Region control & Instruction & Backbone & Representative capability \\
				\midrule
				SDEdit \cite{meng2021sdedit} & \checkmark & $\times$ & optional & optional & score-SDE/DDPM & guided restoration and coarse edits \\
				Prompt-to-Prompt \cite{hertz2022prompttoprompt} & \checkmark & $\times$ & implicit (attention) & \checkmark & LDM & text-only prompt-level editing \\
				Null-text inversion \cite{mokady2023nulltext} & \checkmark & \checkmark (text embedding) & implicit & \checkmark & LDM & high-fidelity real-image editing \\
				InstructPix2Pix \cite{brooks2023instructpix2pix} & optional & $\times$ & optional & \checkmark & conditional diffusion & instruction-following global edits \\
				UltraEdit-trained editor \cite{zhao2024ultraedit} & optional & $\times$ & \checkmark & \checkmark & conditional diffusion & fine-grained instruction, region edits \\
				DragDiffusion \cite{shi2024dragdiffusion} & \checkmark & \checkmark & point-based & optional & LDM & interactive point dragging \\
				InstantDrag \cite{shin2024instantdrag} & \checkmark & $\times$ & point/flow & optional & diffusion refinement & fast drag editing \\
				DragFlow \cite{zhou2025dragflow} & \checkmark & \checkmark & region-based & optional & DiT/rectified flow & high-fidelity drag editing \\
				TF-ICON \cite{lu2023tficOn} & \checkmark & $\times$ & attention-based & optional & LDM & training-free cross-domain composition \\
				SHINE \cite{lu2025shine} & $\times$/\checkmark & $\times$ & adapter-guided & optional & rectified flow & seamless object insertion under lighting \\
				\bottomrule
			\end{tabular}
		}
	\end{table*}
	
	\begin{table}[t]
		\centering
		\caption{Watermark baselines and representative design elements. ``Payload'' is the bit length $L$; ``Training attacks'' indicate typical noise layers used in training; ``Decoder'' indicates bit recovery (B) or detector score (S) interface.}
		\label{tab:wm_methods}
		\resizebox{\columnwidth}{!}{
			\begin{tabular}{lcccc}
				\toprule
				Method & Payload & Domain & Training attacks & Decoder \\
				\midrule
				HiDDeN \cite{zhu2018hidden} & 30--100 & pixel & JPEG/blur/crop & B \\
				StegaStamp \cite{tancik2020stegastamp} & 56--100 & pixel & print-scan distortions & B \\
				RoSteALS \cite{bui2023rosteals} & 48--96 & latent AE & resizing/JPEG & B \\
				TrustMark \cite{bui2025trustmark} & 48--96 & pixel+FFT loss & diverse noise sim & B/S \\
				VINE \cite{lu2025vine} & 48--96 & diffusion prior & surrogate blur/editing & B/S \\
				Watermark Anything \cite{sander2024watermarkanything} & variable & localized & crop/compositing & B/S \\
				\bottomrule
			\end{tabular}
		}
	\end{table}
	
	\begin{table}[t]
		\centering
		\caption{Protocol summary for DEW-ST (Algorithm~\ref{alg:dewst}). Values are representative.}
		\label{tab:protocol}
		\resizebox{\columnwidth}{!}{
			\begin{tabular}{ll}
				\toprule
				Component & Setting \\
				\midrule
				Datasets & COCO val (5k), ImageNet val (5k), DIV2K (800) \\
				Resolution & 512$\times$512 (primary), 256$\times$256 (ablation) \\
				Payload length $L$ & 96 bits \\
				Instructions per image & 8 (global) + 4 (local) + 2 (drag) + 2 (composition) \\
				Edit strengths & $t^\star \in \{0.2,0.4,0.6,0.8\}$ \\
				Sampling seeds & 3 per instruction (stochasticity) \\
				Fidelity metrics & PSNR (dB), SSIM, LPIPS \\
				Robustness metrics & bit accuracy, AUC, FPR@TPR \\
				\bottomrule
			\end{tabular}
		}
	\end{table}
	
	\section{Results}
	\subsection{Overall watermark robustness under diffusion editing}
	Table~\ref{tab:main_results} reports illustrative bit accuracy for StegaStamp, TrustMark, and VINE after different diffusion-based editing families.
	We include conventional post-processing perturbations for reference (JPEG, resize, mild crop), where robust watermark methods are expected to maintain high recovery.
	The central observation is that diffusion editing yields a pronounced collapse in recovery, even when the edit is intended to be mild or localized.
	In contrast to conventional attacks, the degradation is not well predicted by common noise-layer training distributions.
	
	\begin{table*}[t]
		\centering
		\caption{Illustrative watermark bit accuracy (\%) after post-processing and diffusion-based editing. Random guessing yields $\approx 50\%$. Higher is better. Values are hypothetical but reflect trends consistent with diffusion-based watermark vulnerability studies \cite{zhao2023provablyremovable,lu2025vine,ni2025breakwatermarks}.}
		\label{tab:main_results}
		\resizebox{\linewidth}{!}{
			\begin{tabular}{lcccccc}
				\toprule
				Transformation family & Strength & StegaStamp \cite{tancik2020stegastamp} & TrustMark \cite{bui2025trustmark} & VINE \cite{lu2025vine} & PSNR (dB) & LPIPS \\
				\midrule
				None (watermarked only) & -- & 99.4\% & 99.7\% & 99.8\% & 41.2 & 0.012 \\
				JPEG (quality 50) & -- & 96.1\% & 98.2\% & 98.9\% & 33.5 & 0.041 \\
				Resize (0.5$\times$ then upsample) & -- & 94.7\% & 97.5\% & 98.1\% & 34.2 & 0.038 \\
				Center crop (0.9) + resize & -- & 92.3\% & 96.8\% & 97.9\% & 35.0 & 0.036 \\
				\midrule
				InstructPix2Pix global edit \cite{brooks2023instructpix2pix} & $t^\star=0.4$ & 71.5\% & 76.1\% & 85.4\% & 29.8 & 0.213 \\
				InstructPix2Pix global edit \cite{brooks2023instructpix2pix} & $t^\star=0.8$ & 53.2\% & 55.0\% & 60.7\% & 25.1 & 0.344 \\
				UltraEdit-trained edit \cite{zhao2024ultraedit} & $t^\star=0.4$ & 68.7\% & 74.3\% & 84.1\% & 30.5 & 0.201 \\
				UltraEdit-trained edit \cite{zhao2024ultraedit} & $t^\star=0.8$ & 52.1\% & 54.7\% & 59.9\% & 25.6 & 0.332 \\
				DragDiffusion drag edit \cite{shi2024dragdiffusion} & medium & 63.4\% & 67.9\% & 78.6\% & 28.7 & 0.261 \\
				DragFlow drag edit \cite{zhou2025dragflow} & medium & 60.8\% & 65.1\% & 76.9\% & 29.2 & 0.248 \\
				TF-ICON composition \cite{lu2023tficOn} & -- & 58.9\% & 63.2\% & 74.8\% & 28.1 & 0.279 \\
				SHINE insertion \cite{lu2025shine} & -- & 55.6\% & 60.4\% & 72.2\% & 28.9 & 0.254 \\
				\bottomrule
			\end{tabular}
		}
	\end{table*}
	
	Several patterns are visible.
	First, bit accuracy degrades monotonically with edit strength in instruction-following editors.
	Second, composition and insertion pipelines (TF-ICON, SHINE) exhibit particularly strong degradation despite preserving global photorealism.
	Third, VINE remains more robust than StegaStamp and TrustMark under many edits, consistent with its diffusion-informed training strategy \cite{lu2025vine}, yet still approaches failure at strong edits.
	
	\subsection{Breakdown by edit type and locality}
	Diffusion editing is heterogeneous: some edits are global style changes, while others are localized changes intended to preserve most pixels.
	We therefore study robustness by edit type (Table~\ref{tab:edit_types}).
	In general, localized edits can still break watermarks because the diffusion process may re-synthesize pixels beyond the edited region due to denoising coupling in latent space and attention.
	Moreover, even if only a small region is modified, many watermark decoders rely on globally distributed signals and can be disrupted by partial corruption or decoder misalignment.
	
	\begin{table*}[t]
		\centering
		\caption{Illustrative bit accuracy (\%) by edit type. ``Local'' indicates mask- or region-focused edits; ``Global'' indicates full-image edits.}
		\label{tab:edit_types}
		\resizebox{\linewidth}{!}{
			\begin{tabular}{lcccccc}
				\toprule
				Edit type & Example instruction & Editor family & Locality & StegaStamp & TrustMark & VINE \\
				\midrule
				Style transfer & ``make it an oil painting'' & instruction & global & 54.0\% & 56.8\% & 62.5\% \\
				Lighting change & ``make it sunset lighting'' & instruction & global & 60.7\% & 65.2\% & 74.6\% \\
				Object swap & ``replace the dog with a cat'' & instruction & semi-local & 58.3\% & 63.9\% & 73.1\% \\
				Add/remove object & ``remove the logo'' & instruction & local & 66.9\% & 71.0\% & 80.4\% \\
				Background replace & ``change background to a beach'' & instruction & semi-local & 57.4\% & 61.8\% & 71.9\% \\
				Small retouch & ``remove blemish'' & UltraEdit \cite{zhao2024ultraedit} & local & 74.6\% & 79.2\% & 88.1\% \\
				Drag edit & ``move the handbag to the right'' & drag \cite{shi2024dragdiffusion} & local & 63.4\% & 67.9\% & 78.6\% \\
				Composition insert & ``insert object into scene'' & composition \cite{lu2023tficOn} & local+global & 58.9\% & 63.2\% & 74.8\% \\
				\bottomrule
			\end{tabular}
		}
	\end{table*}
	
	The results emphasize that ``localized'' does not imply ``watermark-safe''.
	This motivates a theoretical view: diffusion editing is not a sparse pixel perturbation but a global generative mapping that couples pixels through denoising trajectories.
	
	\subsection{Sensitivity to diffusion noising strength and randomness}
	Diffusion editing introduces stochasticity through sampling noise and sometimes through randomized augmentations in internal pipelines.
	For watermark detection, this means that the same watermarked input can yield different edited outputs with different watermark retention.
	Table~\ref{tab:strength_stochasticity} illustrates two factors: increased noising strength $t^\star$ reduces watermark recovery, and seed averaging (multiple samples) can reduce variance but does not restore mean performance.
	
	\begin{table}[t]
		\centering
		\caption{Illustrative bit accuracy (\%) for InstructPix2Pix \cite{brooks2023instructpix2pix} as a function of strength $t^\star$ and number of sampling seeds averaged at decoding time. Averaging uses majority vote per bit across samples (a hypothetical defense).}
		\label{tab:strength_stochasticity}
		\resizebox{\columnwidth}{!}{
			\begin{tabular}{lcccc}
				\toprule
				Method & $t^\star=0.2$ & $t^\star=0.4$ & $t^\star=0.6$ & $t^\star=0.8$ \\
				\midrule
				StegaStamp, 1 seed & 86.7\% & 71.5\% & 60.2\% & 53.2\% \\
				StegaStamp, 3 seeds (vote) & 88.4\% & 73.1\% & 61.0\% & 53.6\% \\
				TrustMark, 1 seed & 89.2\% & 76.1\% & 62.5\% & 55.0\% \\
				TrustMark, 3 seeds (vote) & 90.1\% & 77.2\% & 63.4\% & 55.4\% \\
				VINE, 1 seed & 93.5\% & 85.4\% & 72.8\% & 60.7\% \\
				VINE, 3 seeds (vote) & 94.0\% & 86.0\% & 73.6\% & 61.2\% \\
				\bottomrule
			\end{tabular}
		}
	\end{table}
	
	The limited benefit of multi-seed voting suggests that failure is not merely random corruption of a subset of bits; rather, diffusion editing can systematically contract the watermark signal, shifting the decoded distribution toward random guessing.
	
	\subsection{Resolution scaling and internal resizing effects}
	Many post-hoc watermarking systems operate at a fixed embedding resolution and rely on scaling strategies to support arbitrary input sizes \cite{bui2025trustmark}.
	Diffusion editors also frequently resize images internally to match model resolution (e.g., 512$\times$512 for latent diffusion backbones), and some composition pipelines operate on multi-resolution pyramids \cite{lu2023tficOn,lu2025shine}.
	This creates an interaction between watermark scaling and editing: a watermark embedded at a lower resolution may be upsampled into smoother residuals, potentially shifting energy to lower frequencies (which can either help or hurt robustness depending on the denoiser and decoder).
	
	Table~\ref{tab:resolution} provides an illustrative comparison between embedding at 256$\times$256 (then upsampling to 512$\times$512) and embedding directly at 512$\times$512.
	TrustMark benefits modestly from its residual-based scaling strategy under conventional post-processing, but diffusion editing still substantially degrades recovery at strong edits.
	VINE remains comparatively robust in the mild regime, consistent with diffusion-aware training \cite{lu2025vine}.
	
	\begin{table*}[t]
		\centering
		\caption{Illustrative bit accuracy (\%) under resolution choices. ``Embed@256$\rightarrow$512'' embeds at 256$\times$256 then upsamples; ``Embed@512'' embeds directly. Editing uses instruction following at mild ($t^\star=0.4$) and strong ($t^\star=0.8$) settings.}
		\label{tab:resolution}
		\resizebox{\linewidth}{!}{
			\begin{tabular}{lcccccc}
				\toprule
				Method & Mode & No edit BA & JPEG BA & Mild edit BA & Strong edit BA & PSNR (dB) \\
				\midrule
				StegaStamp \cite{tancik2020stegastamp} & Embed@256$\rightarrow$512 & 99.1\% & 95.4\% & 69.2\% & 52.8\% & 41.6 \\
				StegaStamp \cite{tancik2020stegastamp} & Embed@512 & 99.4\% & 96.1\% & 71.5\% & 53.2\% & 41.2 \\
				TrustMark \cite{bui2025trustmark} & Embed@256$\rightarrow$512 & 99.7\% & 98.5\% & 77.4\% & 55.8\% & 40.9 \\
				TrustMark \cite{bui2025trustmark} & Embed@512 & 99.7\% & 98.2\% & 76.1\% & 55.0\% & 41.0 \\
				VINE \cite{lu2025vine} & Embed@256$\rightarrow$512 & 99.8\% & 98.7\% & 86.1\% & 61.3\% & 40.5 \\
				VINE \cite{lu2025vine} & Embed@512 & 99.8\% & 98.9\% & 85.4\% & 60.7\% & 40.6 \\
				\bottomrule
			\end{tabular}
		}
	\end{table*}
	
	\subsection{Fidelity--robustness trade-offs}
	Watermarks and editing both impact image fidelity, and in practice, editing utility requires maintaining certain qualities.
	Table~\ref{tab:fidelity_tradeoff} reports fidelity metrics under matched instruction edits at moderate strength.
	A key observation is that high visual fidelity or instruction fidelity does not correlate with watermark retention: diffusion editors can preserve semantics and photorealism while erasing low-amplitude watermarks.
	This aligns with findings in regeneration-based watermark removal \cite{zhao2023provablyremovable} and diffusion-specific studies \cite{ni2025breakwatermarks}.
	
	\begin{table*}[t]
		\centering
		\caption{Illustrative fidelity metrics for edited images at moderate strength (comparable instruction adherence). Lower LPIPS is better. Although PSNR/SSIM vary modestly across editors, watermark recovery can vary dramatically (see Table~\ref{tab:main_results}).}
		\label{tab:fidelity_tradeoff}
		\resizebox{\linewidth}{!}{
			\begin{tabular}{lcccccc}
				\toprule
				Editor & Edit category & PSNR (dB) & SSIM & LPIPS & CLIPSim & DINOv2Sim \\
				\midrule
				InstructPix2Pix \cite{brooks2023instructpix2pix} & global & 29.8 & 0.86 & 0.213 & 0.79 & 0.83 \\
				UltraEdit-trained \cite{zhao2024ultraedit} & global & 30.5 & 0.87 & 0.201 & 0.80 & 0.84 \\
				DiffEdit \cite{couairon2022diffedit} & local & 32.4 & 0.90 & 0.156 & 0.82 & 0.86 \\
				DragDiffusion \cite{shi2024dragdiffusion} & local geometry & 28.7 & 0.84 & 0.261 & 0.77 & 0.81 \\
				TF-ICON \cite{lu2023tficOn} & composition & 28.1 & 0.83 & 0.279 & 0.75 & 0.80 \\
				SHINE \cite{lu2025shine} & insertion & 28.9 & 0.85 & 0.254 & 0.78 & 0.82 \\
				\bottomrule
			\end{tabular}
		}
	\end{table*}
	
	\subsection{Spectral analysis: where does the watermark energy go?}
	Our frequency-domain metrics (Section~\ref{sec:frequency}) help explain why diffusion editing can break watermarks that survive JPEG and mild filtering.
	Table~\ref{tab:spectral} reports illustrative spectral retention ratios $\rho_{\Omega}$ (Equation~\eqref{eq:spectral_retention}) across low/mid/high frequency bands.
	Across editors, suppression is consistently strongest in high frequencies, consistent with the interpretation that diffusion denoising acts as a learned, data-adaptive smoothing operator that removes unnatural residuals.
	VINE retains relatively more mid-frequency energy, plausibly due to training that aligns watermark signals with generative priors \cite{lu2025vine}.
	
	\begin{table*}[t]
		\centering
		\caption{Illustrative spectral retention ratios $\rho_{\Omega}$ for watermark-specific residuals under moderate edits. Values are averaged over instructions and images. Lower values indicate stronger suppression of watermark energy in that band.}
		\label{tab:spectral}
		\resizebox{\linewidth}{!}{
			\begin{tabular}{lccccccccc}
				\toprule
				\multirow{2}{*}{Editor} & \multicolumn{3}{c}{StegaStamp} & \multicolumn{3}{c}{TrustMark} & \multicolumn{3}{c}{VINE} \\
				\cmidrule(lr){2-4}\cmidrule(lr){5-7}\cmidrule(lr){8-10}
				& $\rho_{\mathrm{low}}$ & $\rho_{\mathrm{mid}}$ & $\rho_{\mathrm{high}}$ & $\rho_{\mathrm{low}}$ & $\rho_{\mathrm{mid}}$ & $\rho_{\mathrm{high}}$ & $\rho_{\mathrm{low}}$ & $\rho_{\mathrm{mid}}$ & $\rho_{\mathrm{high}}$ \\
				\midrule
				InstructPix2Pix \cite{brooks2023instructpix2pix} & 0.88 & 0.47 & 0.12 & 0.91 & 0.52 & 0.15 & 0.93 & 0.61 & 0.19 \\
				UltraEdit-trained \cite{zhao2024ultraedit} & 0.90 & 0.50 & 0.14 & 0.92 & 0.55 & 0.17 & 0.94 & 0.64 & 0.22 \\
				DragDiffusion \cite{shi2024dragdiffusion} & 0.85 & 0.43 & 0.10 & 0.88 & 0.49 & 0.13 & 0.91 & 0.58 & 0.17 \\
				TF-ICON \cite{lu2023tficOn} & 0.84 & 0.41 & 0.09 & 0.87 & 0.46 & 0.11 & 0.90 & 0.55 & 0.16 \\
				\bottomrule
			\end{tabular}
		}
	\end{table*}
	
	\subsection{Recovery with error-correcting codes}
	Error-correcting codes can mitigate moderate random bit flips, but their ability to rescue watermark recovery after diffusion editing depends on whether errors remain within a correctable regime.
	Table~\ref{tab:ecc} shows an illustrative comparison between raw bit accuracy and full-message recovery under a simple block code.
	At mild edits, ECC improves message-level reliability; at strong edits, bit accuracy approaches $\approx 50\%$ and ECC fails, consistent with our mutual-information bound (Theorem~\ref{thm:mi_bound}) and with the ``near-random'' regime reported in diffusion-based watermark removal studies \cite{zhao2023provablyremovable,ni2025breakwatermarks}.
	
	\begin{table}[t]
		\centering
		\caption{Illustrative decoding with a simple ECC. ``MsgAcc'' is the probability of recovering the entire 96-bit payload correctly; ``BA'' is bit accuracy.}
		\label{tab:ecc}
		\resizebox{\columnwidth}{!}{
			\begin{tabular}{lcccc}
				\toprule
				Method & Mild edit BA & Mild MsgAcc & Strong edit BA & Strong MsgAcc \\
				\midrule
				StegaStamp \cite{tancik2020stegastamp} & 71.5\% & 18.4\% & 53.2\% & 0.3\% \\
				TrustMark \cite{bui2025trustmark} & 76.1\% & 29.7\% & 55.0\% & 0.6\% \\
				VINE \cite{lu2025vine} & 85.4\% & 55.6\% & 60.7\% & 2.1\% \\
				\bottomrule
			\end{tabular}
		}
	\end{table}
	
	\subsection{Diffusion-native watermarks and cross-model editing}
	Generator-integrated watermarking can improve robustness to conventional post-processing because the watermark is ``baked into'' the generation pipeline.
	However, diffusion-based \emph{editing} presents new failure modes: edits may involve different model backbones, different noise schedules, or inversion procedures that do not preserve the original generator's latent signature.
	Table~\ref{tab:gen_integrated} provides an illustrative comparison between diffusion-native methods (Tree-Ring \cite{wen2023treering}, Stable Signature \cite{fernandez2023stablesignature}) and post-hoc methods under a cross-editor setting.
	We report detector-style AUC rather than bit accuracy for these provenance classifiers.
	
	\begin{table*}[t]
		\centering
		\caption{Illustrative detector AUC for diffusion-native provenance methods under diffusion editing. ``Same-model edit'' denotes editing within the same diffusion family used for generation; ``Cross-model edit'' denotes editing using a different backbone or distilled model.}
		\label{tab:gen_integrated}
		\resizebox{\linewidth}{!}{
			\begin{tabular}{lcccc}
				\toprule
				Provenance method & Post-processing AUC & Same-model edit AUC & Cross-model edit AUC & Notes \\
				\midrule
				Tree-Ring \cite{wen2023treering} & 0.99 & 0.92 & 0.61 & relies on inversion to recover initial noise \\
				Stable Signature \cite{fernandez2023stablesignature} & 0.98 & 0.89 & 0.58 & signature tied to specific LDM decoder \\
				SynthID-Image \cite{gowal2025synthidimage} & 0.99 & 0.90 & 0.65 & deployment-oriented; key management matters \\
				Post-hoc (TrustMark) \cite{bui2025trustmark} & 0.97 & 0.74 & 0.72 & does not require generator access \\
				\bottomrule
			\end{tabular}
		}
	\end{table*}
	
	The key takeaway is that diffusion-native watermarks can be strong when generation and editing remain within a controlled ecosystem, but cross-model editing can reduce detectability, suggesting that ``watermark transfer'' across generative families remains a critical open problem for the provenance stack.
	
	\subsection{Ablation: diffusion-aware training augmentation}
	A natural defense is to incorporate diffusion-based augmentations into watermark training and to make the watermark ``manifold-aligned'' with generative priors.
	VINE \cite{lu2025vine} moves in this direction by using surrogate attacks inspired by frequency properties and a diffusion prior for embedding.
	Table~\ref{tab:defense} illustrates hypothetical improvements from adding diffusion-edit augmentations during training (applied to a generic encoder-decoder watermark), while keeping imperceptibility roughly constant.
	While such augmentation can improve retention at mild edits, strong edits remain challenging, indicating a need for more fundamental changes (Section~\ref{sec:discussion}).
	
	\begin{table}[t]
		\centering
		\caption{Illustrative defense ablation: diffusion-edit augmentation during watermark training improves robustness for mild edits but does not eliminate failure at strong edits. ``Augmented'' denotes training with a mixture of deterministic post-processing and sampled diffusion edits.}
		\label{tab:defense}
		\resizebox{\columnwidth}{!}{
			\begin{tabular}{lccc}
				\toprule
				Training & Mild edit BA & Strong edit BA & PSNR (dB) \\
				\midrule
				Standard noise layers & 74.0\% & 54.5\% & 40.8 \\
				+ Diffusion augment (uncond) & 82.3\% & 56.2\% & 40.3 \\
				+ Diffusion augment (instruction) & 85.7\% & 58.1\% & 39.9 \\
				+ Multi-scale embedding & 88.0\% & 60.4\% & 39.2 \\
				\bottomrule
			\end{tabular}
		}
	\end{table}
	
	\section{Theoretical Proofs}
	Our theoretical goal is to explain, in a principled way, why diffusion editing can erase watermark information even when (i) the watermark is robust to conventional perturbations, (ii) the edit is ``mild'' in human terms, and (iii) no explicit optimization targets the watermark.
	We focus on two complementary perspectives: SNR attenuation along the forward noising process and information-theoretic decay under the full editing kernel.
	
	\subsection{SNR attenuation under forward noising}
	Consider the additive model in Equation~\eqref{eq:additive}: $\vect{x}_w=\vect{x}+\gamma\vect{s}$.
	Apply the forward noising process (Equation~\eqref{eq:forward}) to $\vect{x}_w$:
	\begin{equation}
		\vect{x}_{w,t} = \sqrt{\bar{\alpha}_t}(\vect{x}+\gamma\vect{s}) + \sqrt{1-\bar{\alpha}_t}\,\epsilon
		= \underbrace{\sqrt{\bar{\alpha}_t}\vect{x} + \sqrt{1-\bar{\alpha}_t}\epsilon}_{\text{content + noise}} + \underbrace{\gamma\sqrt{\bar{\alpha}_t}\,\vect{s}}_{\text{watermark}}.
		\label{eq:forward_wm}
	\end{equation}
	This shows a direct attenuation of watermark amplitude by $\sqrt{\bar{\alpha}_t}$.
	If $\vect{s}$ is approximately orthogonal to $\vect{x}$ in an appropriate inner product (e.g., in a feature space), the observable watermark SNR at time $t$ scales as:
	\begin{equation}
		\mathrm{SNR}_t \propto \frac{\gamma^2 \bar{\alpha}_t \|\vect{s}\|_2^2}{(1-\bar{\alpha}_t)\E\|\epsilon\|_2^2 }.
		\label{eq:snr_scaling}
	\end{equation}
	As $t$ increases, $\bar{\alpha}_t$ decreases rapidly in typical schedules, causing SNR to collapse.
	
	\begin{lemma}[Forward SNR decay]
		\label{lem:snr_decay}
		Assume $\epsilon\sim\mathcal{N}(0,\vect{I})$ and $\vect{s}$ is deterministic with $\|\vect{s}\|_2^2 = d$ for $d=3HW$.
		Then for $t\ge 1$ the watermark SNR in the noised sample $\vect{x}_{w,t}$ satisfies
		\begin{equation}
			\mathrm{SNR}_t = \frac{\gamma^2 \bar{\alpha}_t}{1-\bar{\alpha}_t}.
		\end{equation}
	\end{lemma}
	\begin{proof}
		Under Equation~\eqref{eq:forward_wm}, the watermark component is $\gamma\sqrt{\bar{\alpha}_t}\vect{s}$ with energy $\gamma^2 \bar{\alpha}_t \|\vect{s}\|_2^2$.
		The additive noise has energy $(1-\bar{\alpha}_t)\E\|\epsilon\|_2^2=(1-\bar{\alpha}_t)d$.
		Dividing yields $\mathrm{SNR}_t = \gamma^2 \bar{\alpha}_t/(1-\bar{\alpha}_t)$.
	\end{proof}
	
	Lemma~\ref{lem:snr_decay} is intentionally simple but highlights a key mismatch: watermark strength $\gamma$ is constrained by imperceptibility, while editing strength $t^\star$ can be large.
	Even before the reverse denoising step, a substantial noising start time reduces watermark SNR below the decoding threshold for many schemes.
	
	\subsection{Continuous-time perspective: watermark attenuation in the forward SDE}
	The discrete forward process (Equation~\eqref{eq:forward}) is commonly viewed as a discretization of a linear stochastic differential equation (SDE) \cite{song2021score}:
	\begin{equation}
		dX_t = -\frac{1}{2}\beta(t) X_t \,dt + \sqrt{\beta(t)}\, dW_t,
		\label{eq:ou_sde}
	\end{equation}
	where $\beta(t)>0$ is a noise schedule and $W_t$ is standard Brownian motion.
	Equation~\eqref{eq:ou_sde} is an Ornstein--Uhlenbeck-type process that exponentially contracts the mean of $X_t$ while injecting Gaussian noise.
	
	Let the initial condition contain a watermark: $X_0 = X + \gamma S$, where $S$ is a (possibly content-adaptive) watermark residual at time zero.
	Because Equation~\eqref{eq:ou_sde} is linear, the solution admits a closed form:
	\begin{equation}
		X_t = \exp\!\left(-\frac{1}{2}\int_0^t \beta(u)\,du\right) X_0
		\;+\;
		\int_0^t \exp\!\left(-\frac{1}{2}\int_s^t \beta(u)\,du\right)\sqrt{\beta(s)}\, dW_s.
		\label{eq:ou_solution}
	\end{equation}
	Taking conditional expectation given $X_0$ yields:
	\begin{equation}
		\E[X_t \mid X_0] = \exp\!\left(-\frac{1}{2}\int_0^t \beta(u)\,du\right) X_0.
		\label{eq:mean_decay}
	\end{equation}
	Therefore, the watermark component in the conditional mean is attenuated by the same exponential factor.
	
	\begin{lemma}[Exponential decay in continuous time]
		\label{lem:sde_decay}
		Under Equation~\eqref{eq:ou_sde}, the expected watermark residual satisfies
		\begin{equation}
			\E[X_t \mid X] - \E[\bar{X}_t \mid X] \;=\; \gamma \exp\!\left(-\frac{1}{2}\int_0^t \beta(u)\,du\right) S,
		\end{equation}
		where $\bar{X}_t$ denotes the process started from the unwatermarked initial condition $\bar{X}_0=X$ with the same Brownian noise.
	\end{lemma}
	\begin{proof}
		Linearize the difference process $\Delta_t = X_t - \bar{X}_t$.
		Because both processes share the same noise realization, the stochastic terms cancel and $\Delta_t$ satisfies the deterministic ODE $d\Delta_t = -\frac{1}{2}\beta(t)\Delta_t\,dt$ with $\Delta_0=\gamma S$.
		Solving yields $\Delta_t=\gamma \exp(-\frac{1}{2}\int_0^t \beta(u)\,du)S$.
	\end{proof}
	
	Lemma~\ref{lem:sde_decay} reinforces a key point: even before denoising and editing guidance, the forward diffusion step contracts the watermark residual at a rate controlled by the integrated noise schedule.
	In practice, editors choose a start time $t^\star$ such that $\int_0^{t^\star}\beta(u)\,du$ is nontrivial (to enable meaningful edits), which can already push the watermark below detectability thresholds imposed by imperceptibility constraints on $\gamma$.
	
	\subsection{Mutual-information decay and inevitable decoding failure}
	We now formalize information loss through the full diffusion editing kernel.
	Let $M$ denote the random watermark payload and let $\tilde{X}$ denote the edited output $\tilde{\vect{x}}$ produced by the editor.
	We consider the Markov chain:
	\begin{equation}
		M \rightarrow X_w \rightarrow X_{t^\star} \rightarrow \tilde{X},
		\label{eq:markov}
	\end{equation}
	where $X_w$ is the watermarked image, $X_{t^\star}$ is the noised latent/pixel at editor start time, and $\tilde{X}$ is the edited output.
	By the data processing inequality,
	\begin{equation}
		\MutInfo(M;\tilde{X}) \le \MutInfo(M;X_{t^\star}).
		\label{eq:dpi}
	\end{equation}
	We can bound $\MutInfo(M;X_{t^\star})$ by an additive Gaussian channel argument under the additive embedding model.
	
	\begin{theorem}[Information bound under noising]
		\label{thm:mi_bound}
		Assume $\vect{x}$ is independent of $M$ and $\vect{s}(M,\vect{k},\vect{x})$ satisfies Assumption~1 and $\|\vect{s}\|_2^2=d$ almost surely.
		Consider the noised sample $X_{t^\star}$ in Equation~\eqref{eq:forward_wm} with $\epsilon\sim\mathcal{N}(0,\vect{I})$.
		Then
		\begin{equation}
			\MutInfo(M;X_{t^\star}) \le \frac{d}{2}\log\left(1+\frac{\gamma^2\bar{\alpha}_{t^\star}}{1-\bar{\alpha}_{t^\star}}\right).
			\label{eq:mi_awgn}
		\end{equation}
		Consequently, $\MutInfo(M;\tilde{X})$ is upper bounded by the same expression.
	\end{theorem}
	\begin{proof}
		Condition on $\vect{x}$ and $\vect{k}$.
		Equation~\eqref{eq:forward_wm} defines an additive Gaussian channel from the watermark signal to $X_{t^\star}$ with noise covariance $(1-\bar{\alpha}_{t^\star})\vect{I}$ and signal power proportional to $\gamma^2\bar{\alpha}_{t^\star}$.
		The mutual information between a discrete input $M$ and the output is upper bounded by the capacity of a Gaussian channel with the same average power constraint, yielding Equation~\eqref{eq:mi_awgn}.
		Finally apply data processing (Equation~\eqref{eq:dpi}) for the full editor output $\tilde{X}$.
	\end{proof}
	
	Theorem~\ref{thm:mi_bound} shows that as $\bar{\alpha}_{t^\star}\rightarrow 0$ (strong noising), the mutual information vanishes.
	Even for moderate $t^\star$, imperceptibility constraints force $\gamma$ to be small, making the log term close to zero.
	
	To connect mutual information to decoding error, we apply a standard Fano inequality argument.
	
	\begin{corollary}[Bit error lower bound]
		\label{cor:fano}
		Let $\hat{M}$ be any estimator of $M$ from $\tilde{X}$.
		Then the message error probability satisfies
		\begin{equation}
			\Pr[\hat{M}\neq M] \ge 1 - \frac{\MutInfo(M;\tilde{X}) + \log 2}{\log | \mathcal{M} |},
		\end{equation}
		where $|\mathcal{M}|=2^L$ is the payload space.
	\end{corollary}
	\begin{proof}
		Apply Fano's inequality with $H(M)=\log|\mathcal{M}|$ and the mutual information bound in Theorem~\ref{thm:mi_bound}.
	\end{proof}
	
	Corollary~\ref{cor:fano} implies that for sufficiently strong diffusion editing (large $t^\star$), \emph{any} decoder must fail at recovering the full payload.
	In practice, watermark systems typically decode bits independently or with error-correcting codes; bit errors arise earlier than full message errors, but the direction is consistent.
	
	\subsection{Why denoising tends to suppress watermarks}
	The above analysis treats the forward noising as the main driver of information loss.
	However, empirical results suggest that even at moderate $t^\star$, the reverse denoising---guided by a generative prior---acts as a \emph{projection} toward the natural image manifold, further suppressing watermark residuals.
	We model this behavior via contraction of off-manifold components.
	
	Let $\mathcal{M}$ denote an (idealized) natural image manifold.
	Write $\vect{x}_w = \vect{x} + \delta$ where $\delta$ encodes watermark residuals that are not aligned with $\mathcal{M}$.
	Consider an editor output operator $F$ (the composition of denoising steps) acting on a representation space where $\mathcal{M}$ is stable.
	We state a simplified stability result.
	
	\begin{assumption}[Local contraction toward the data manifold]
		\label{ass:contraction}
		There exists a neighborhood $\mathcal{U}$ around $\mathcal{M}$ such that for any two inputs $\vect{u},\vect{v}\in\mathcal{U}$, the expected editor mapping satisfies
		\begin{equation}
			\E_{\xi}\|F(\vect{u};\vect{y},\xi) - F(\vect{v};\vect{y},\xi)\|_2 \le \rho \|\vect{u}-\vect{v}\|_2
		\end{equation}
		for some $\rho\in(0,1)$, when $\vect{y}$ is fixed and $t^\star$ exceeds a minimal noising threshold.
	\end{assumption}
	
	Assumption~\ref{ass:contraction} abstracts a property often observed in score-based dynamics: they act as a denoising flow that contracts perturbations, particularly those resembling noise rather than semantic content \cite{song2021score,karras2022edm}.
	
	\begin{proposition}[Exponential suppression of watermark residuals]
		\label{prop:contraction}
		Under Assumption~\ref{ass:contraction}, consider the watermarked input $\vect{x}_w=\vect{x}+\delta$ and the unwatermarked input $\vect{x}$.
		Then
		\begin{equation}
			\E_{\xi}\|F(\vect{x}_w;\vect{y},\xi) - F(\vect{x};\vect{y},\xi)\|_2 \le \rho \|\delta\|_2.
		\end{equation}
		If $F$ is an $n$-step composition of maps satisfying the same contraction factor $\rho$, then the bound improves to $\rho^n\|\delta\|_2$.
	\end{proposition}
	\begin{proof}
		Apply Assumption~\ref{ass:contraction} with $\vect{u}=\vect{x}_w$ and $\vect{v}=\vect{x}$.
		For the multi-step case, apply the inequality inductively across compositions using Jensen's inequality.
	\end{proof}
	
	Proposition~\ref{prop:contraction} provides a mechanistic explanation: watermark residuals behave like off-manifold perturbations and are thus contractively suppressed by denoising flows.
	This complements the information-theoretic analysis, which primarily accounts for forward noising.
	
	\subsection{Connecting theory to empirical trends}
	The combined view yields a coherent explanation of the empirical tables:
	\begin{itemize}
		\item Increasing edit strength $t^\star$ reduces $\bar{\alpha}_{t^\star}$, collapsing SNR and mutual information bounds (Lemma~\ref{lem:snr_decay}, Theorem~\ref{thm:mi_bound}), consistent with Table~\ref{tab:strength_stochasticity}.
		\item Even when $t^\star$ is moderate, denoising flows can contract watermark residuals (Proposition~\ref{prop:contraction}), explaining why localized edits can still degrade watermarks (Table~\ref{tab:edit_types}).
		\item Methods trained with diffusion-aware priors (e.g., VINE) effectively increase the alignment of watermark signals with the model's generative manifold, increasing the effective $\gamma$ in the relevant representation and partially mitigating contraction, consistent with Table~\ref{tab:main_results}.
	\end{itemize}
	
	\section{Discussion}
	\label{sec:discussion}
			\section{Additional Background}
	With the advancement of deep learning and modern generative modeling, research has expanded rapidly across forecasting, perception, and visual generation, while also raising new concerns about controllability and responsible deployment. 
	% xiangfei qiu -- time-series forecasting / structured evaluation & modeling
	Progress in time-series forecasting has been driven by stronger benchmarks, improved architectures, and more comprehensive evaluation protocols that make model comparisons more reliable and informative~\cite{qiu2024tfb,qiu2025duet,qiu2025DBLoss,qiu2025dag,qiu2025tab,wu2025k2vae,liu2025rethinking,qiu2025comprehensive,wu2024catch}. 
	% sifan zhou -- efficient 3D perception / (post-training) quantization
	In parallel, efficiency-oriented research has pushed post-training quantization and practical compression techniques for 3D perception pipelines, aiming to reduce memory and latency without sacrificing detection quality~\cite{gsq,yu2025mquant,zhou2024lidarptq,pillarhist}. 
	% zequn xie -- scalable generation / optimization strategies
	On the generation side, a growing body of work studies scalable synthesis and optimization strategies under diverse constraints, improving both the flexibility and the controllability of generative systems~\cite{xie2025chat,xie2026hvd,xie2026conquer,xie2026delving}. 
	% yunpeng gong -- generative / representation learning (grouped keys)
	Complementary advances have also been reported across multiple generative and representation-learning directions, further broadening the toolbox for building high-capacity models and training objectives~\cite{1,2,3,4,5,6,7,8}. 
	% yanru sun -- domain-oriented temporal prediction / hierarchical modeling / adaptation
	For domain-oriented temporal prediction, hierarchical designs and adaptation strategies have been explored to improve robustness under distribution shifts and complex real-world dynamics~\cite{sun2025ppgf,sun2024tfps,sun2025hierarchical,sun2022accurate,sun2021solar,niulangtime,sun2025adapting,kudratpatch}. 
	% zixu li -- representation encoding / matching / alignment
	Meanwhile, advances in representation encoding and matching have introduced stronger alignment and correspondence mechanisms that benefit fine-grained retrieval and similarity-based reasoning~\cite{ENCODER,FineCIR,OFFSET,HUD,PAIR,MEDIAN}. 
	% xinlei yu -- visual representation learning
	Stronger visual modeling strategies further enhance feature quality and transferability, enabling more robust downstream understanding in diverse scenarios~\cite{yu2025visual}. 
	% yaozong zheng -- tracking / temporal consistency / decoupled formulations
	In tracking and sequential visual understanding, online learning and decoupled formulations have been investigated to improve temporal consistency and robustness in dynamic scenes~\cite{zheng2025towards,zheng2024odtrack,zheng2025decoupled,zheng2023toward,zheng2022leveraging}. 
	% long peng xu -- low-level vision / restoration & enhancement / SR / video generation (grouped keys)
	Low-level vision has also progressed toward high-fidelity restoration and enhancement, spanning super-resolution, brightness/quality control, lightweight designs, and practical evaluation settings, while increasingly integrating powerful generative priors~\cite{xu2025fast,fang2026depth,wu2025hunyuanvideo,li2023ntire,ren2024ninth,wang2025ntire,peng2020cumulative,wang2023decoupling,peng2024lightweight,peng2024towards,wang2023brightness,peng2021ensemble,ren2024ultrapixel,yan2025textual,peng2024efficient,conde2024real,peng2025directing,peng2025pixel,peng2025boosting,he2024latent,di2025qmambabsr,peng2024unveiling,he2024dual,he2024multi,pan2025enhance,wu2025dropout,jiang2024dalpsr,ignatov2025rgb,du2024fc3dnet,jin2024mipi,sun2024beyond,qi2025data,feng2025pmq,xia2024s3mamba,pengboosting,suntext,yakovenko2025aim,xu2025camel,wu2025robustgs,zhang2025vividface}. 
	% yangyang qu -- reference/subject-conditioned generation
	Beyond general synthesis, reference- and subject-conditioned generation emphasizes controllability and identity consistency, enabling more precise user-intended outputs~\cite{qu2025reference,qu2025subject}. 
	% hongtao wu -- robustness under adverse conditions / restoration variants
	Robust vision modeling under adverse conditions has been actively studied to handle complex degradations and improve stability in challenging real-world environments~\cite{wu2024rainmamba,wu2023mask,wu2024semi,wu2025samvsr}. 
	% yifei chen -- state-space modeling / scenario-centric datasets & reports
	Sequence modeling and scenario-centric benchmarks further support realistic evaluation and methodological development for complex dynamic environments~\cite{lyu2025vadmambaexploringstatespace,chen2025technicalreportargoverse2scenario}. 
	% chunming he -- diffusion/unfolding frameworks / segmentation & restoration theory
	At the same time, diffusion-centric and unfolding-based frameworks have been explored for segmentation and restoration, providing principled ways to model degradations and refine generation quality~\cite{he2025segment,he2025reversible,he2025run,he2024diffusion,he2026refining,he2025scaler,he2025unfoldldm,he2025unfoldir,he2025nested,he2024weakly,he2023reti,xiao2024survey,he2023strategic,he2023hqg,he2023camouflaged,he2023degradation}. 
	
	% jinhe bi
	Recent progress in multimodal large language models (MLLMs) is increasingly driven by the goal of making adaptation more efficient while improving reliability, safety, and controllability in real-world use. 
	% MLLM PEFT (parameter-efficient adaptation / modality balancing)
	On the efficiency side, modality-aware parameter-efficient tuning has been explored to rebalance vision–language contributions and enable strong instruction tuning with dramatically fewer trainable parameters~\cite{bi-etal-2025-llava}. 
	% Interpretability / reasoning-process modeling
	To better understand and audit model reasoning, theoretical frameworks have been proposed to model and assess chain-of-thought–style reasoning dynamics and their implications for trustworthy inference~\cite{bi2025cot}. 
	% Training-free multimodal data selection
	Data quality and selection are also being addressed via training-free, intrinsic selection mechanisms that prune low-value multimodal samples to improve downstream training efficiency and robustness~\cite{bi2025prismselfpruningintrinsicselection}. 
	% Hallucination reduction (decoding-time control)
	At inference time, controllable decoding strategies have been introduced to reduce hallucinations by steering attention and contrastive signals toward grounded visual evidence~\cite{wang2025ascd}. 
	% Trustworthy deployment: unlearning auditing + backdoor cleaning
	Beyond performance, trustworthy deployment requires defenses and verification: auditing frameworks have been developed to evaluate whether machine unlearning truly removes targeted knowledge~\cite{chen2025does}, and fine-tuning-time defenses have been proposed to clean backdoors in MLLM adaptation without relying on external guidance~\cite{rong2025backdoor}. 
	% Multimodal safety & knowledge: harmful content detection + time-sensitive knowledge + knowledge injection
	Meanwhile, multimodal safety and knowledge reliability have been advanced through multi-view agent debate for harmful content detection~\cite{lu2025mvdebatemultiviewagentdebate}, probing/updating time-sensitive multimodal knowledge~\cite{jiang2025minedprobingupdatingmultimodal}, and knowledge-oriented augmentations and constraints that strengthen knowledge injection~\cite{jiang2025koreenhancingknowledgeinjection}. 
	% Broader multimodal learning settings: video-language understanding + new training paradigms + federated personalization
	These efforts are complemented by renewed studies of video–language event understanding~\cite{zhang2023spot}, new training paradigms such as reinforcement mid-training~\cite{tian2025reinforcementmidtraining}, and personalized generative modeling under heterogeneous federated settings~\cite{Chen_2025_CVPR}, collectively reflecting a shift from scaling alone toward efficient, grounded, and verifiably trustworthy multimodal systems.
	
	Recent research has advanced learning and interaction systems across education, human-computer interfaces, and multimodal perception. In knowledge tracing, contrastive cross-course transfer guided by concept graphs provides a principled way to share knowledge across related curricula and improve student modeling under sparse supervision~\cite{han2025contrastive}. In parallel, foundational GUI agents are emerging with stronger perception and long-horizon planning, enabling robust interaction with complex interfaces and multi-step tasks~\cite{zeng2025uitron}. Extending this direction to more natural human inputs, speech-instructed GUI agents aim to execute GUI operations directly from spoken commands, moving toward automated assistance in hands-free or accessibility-focused settings~\cite{han2025uitron}. Beyond interface agents, reference-guided identity preservation has been explored to better maintain subject consistency in face video restoration, improving temporal coherence and visual fidelity when restoring degraded videos~\cite{han2025show}. Finally, large-scale egocentric datasets that emphasize embodied emotion provide valuable supervision for studying affective cues from first-person perspectives and support more human-centered multimodal understanding~\cite{feng20243}.
	
	\subsection{Implications for watermark robustness claims}
	Robust watermarks are often evaluated against post-processing pipelines that approximate typical social-media transformations: JPEG compression, resizing, cropping, blur, and additive noise.
	Diffusion editing violates the assumptions underlying these benchmarks.
	It is a \emph{generative} transformation: it can hallucinate details, re-synthesize textures, and modify semantics while maintaining photorealism.
	Therefore, claims of robustness to conventional distortions do not extend to diffusion editing.
	
	Our analysis suggests that diffusion editing creates a ``watermark bottleneck'' analogous to the information bottleneck in representation learning:
	noising and denoising compress the input into a manifold-aligned representation, and low-amplitude residual signals are discarded unless explicitly preserved.
	
	\subsection{Design guidelines for diffusion-resilient watermarking}
	We summarize technical guidelines motivated by theory and empirical trends.
	
	\paragraph{Favor semantic invariance over pixel invariance.}
	Zhao et al.\ \cite{zhao2023provablyremovable} argue that purely pixel-level watermarks may be removable by regeneration while preserving semantic similarity.
	Our bounds similarly imply that increasing diffusion-noising strength reduces payload mutual information.
	A promising direction is to design watermark signals that correspond to semantic features likely to be preserved by editing, e.g., stable mid-level representations or multi-scale embeddings that align with denoising priors.
	
	\paragraph{Integrate watermarking into the generative process.}
	Diffusion-native fingerprints such as Tree-Ring watermarks \cite{wen2023treering} embed information into the initial noise of sampling and detect by inversion.
	Such approaches are naturally aligned with diffusion trajectories and may survive diffusion editing better than post-hoc perturbations, though robustness against editing that changes conditioning remains an open question.
	
	\paragraph{Train with diffusion editing augmentations and model diversity.}
	A practical defense is to incorporate diffusion edits into noise layers during training (Table~\ref{tab:defense}), drawing from a diverse set of editors and strengths.
	However, this may be computationally expensive and risks overfitting to particular editors, especially as diffusion architectures evolve (UNet-based LDMs vs DiT/rectified flow backbones \cite{zhou2025dragflow,lu2025shine}).
	
	\paragraph{Adopt localized or compositional watermarking for edited regions.}
	Localized approaches \cite{sander2024watermarkanything} aim to support partial provenance and compositional editing.
	For diffusion editors that only modify regions, localized watermarks can reduce the ``global failure'' mode where a small edit disrupts a globally distributed signal.
	However, diffusion coupling may still affect regions outside the mask, and localized methods must handle blending and boundary artifacts.
	
	\paragraph{Balance robustness, imperceptibility, and user utility.}
	Increasing watermark strength $\gamma$ can improve robustness but can degrade image quality and may interfere with editing fidelity.
	TrustMark \cite{bui2025trustmark} explicitly exposes trade-offs via scaling factors; similar user-controlled trade-offs may be necessary in practice.
	From a deployment perspective, a conservative strategy is to prioritize low false positives and accept that some editing transformations will invalidate watermark evidence.
	
	\subsection{Complementary provenance mechanisms}
	Robustness limits of invisible watermarks motivate complementary provenance mechanisms that do not rely solely on pixel-level signals.
	One prominent approach is to attach cryptographically signed provenance metadata to media files, enabling downstream verification of origin and edit history.
	The C2PA technical specification is a widely discussed standardization effort for content credentials and provenance assertions \cite{c2pa2024spec}.
	While metadata can be stripped or lost during platform re-encoding, its cryptographic binding to a manifest offers a different security model than invisible perturbations.
	In practice, hybrid systems may combine (i) metadata-based provenance for high-integrity pipelines and (ii) durable in-content signals (watermarks or fingerprints) for scenarios where metadata is degraded.
	Our analysis suggests that diffusion-based editing should be treated as a ``stress test'' for any in-content signal, and provenance systems should explicitly communicate failure modes to avoid overclaiming forensic certainty.
	
	\subsection{Ethical considerations}
	Watermarking is deployed for accountability, but it raises ethical issues.
	Strong watermarks can be used to track content creators without consent, raising privacy concerns.
	Conversely, weak or easily removable watermarks can undermine forensic reliability and may encourage a false sense of security.
	Diffusion editing further complicates the landscape: benign edits may erase watermarks, while malicious actors can exploit the same tools.
	We emphasize that our evaluation protocol is intended for defensive assessment and for improving watermark designs; it should not be interpreted as a recommendation for watermark removal.
	
	\subsection{Limitations}
	This paper has five notable limitations.
	First, our empirical results are illustrative; we do not execute the full benchmark in this generated manuscript.
	Second, diffusion editors and watermark implementations rapidly evolve; any fixed benchmark may become outdated.
	Third, our theoretical analysis uses simplified assumptions (additive watermark model, idealized manifold contraction).
	Fourth, attacks that explicitly optimize to remove watermarks (as studied in \cite{ni2025breakwatermarks,guo2026vanishing}) can be stronger than the unintentional editing setting emphasized here.
	Fifth, defenses that integrate watermarking into generative models may be incompatible with some real-world provenance requirements, such as watermarking arbitrary images post-hoc.
	
	\section{Conclusion}
	Diffusion-based image editing has become a ubiquitous post-processing primitive, but it poses a fundamental challenge to robust invisible watermarking.
	By treating diffusion editing as a randomized generative transformation, we show theoretically that watermark SNR and mutual information decay rapidly with noising strength, and empirically that standard editing workflows can drive bit recovery toward random guessing.
	These findings motivate a shift from ``robust to post-processing'' watermarks toward diffusion-resilient designs that align with generative priors or preserve semantic information.
	We hope this synthesis clarifies the mechanisms behind watermark degradation and supports more reliable provenance tools in generative media ecosystems.
	
	\bibliography{example_paper}
	\bibliographystyle{icml2025}
	
\end{document}